\documentstyle[11pt,epsf]{article}

\begin{document}

\newtheorem{theorem}{Theorem}[section]
\newtheorem{lemma}[theorem]{Lemma}
\newtheorem{fact}[theorem]{Fact}
\newtheorem{corollary}[theorem]{Corollary}
\newtheorem{definition}{Definition}[section]
\newtheorem{proposition}[theorem]{Proposition}
\newtheorem{observation}[theorem]{Observation}
\newtheorem{claim}[theorem]{Claim}
\newtheorem{assumption}[theorem]{Assumption}
\newtheorem{notation}[theorem]{Notation}
\newenvironment{proof}{\noindent{\bf Proof:\/}}{\hfill $\Box$\vskip 0.1in}
\newenvironment{proofsp}{\noindent{\bf Proof}}{\hfill $\Box$\vskip 0.1in}
\pagenumbering{arabic}

\newcommand {\ignore} [1] {}

\date{}

\pagenumbering{arabic}

\title{Approximating subset $k$-connectivity problems}

\author{Zeev Nutov \\
\small The Open University of Israel \\ 
\small {\tt nutov@openu.ac.il}}

\maketitle

\begin{abstract}
A subset $T \subseteq V$ of terminals is $k$-connected to a root $s$ in a 
directed/undirected graph $J$ if $J$ has $k$ internally-disjoint $vs$-paths 
for every $v \in T$; $T$ is $k$-connected in $J$ if $T$ is $k$-connected to every $s \in T$.
We consider the {\sf Subset $k$-Connectivity Augmentation} problem:
given a graph $G=(V,E)$ with edge/node-costs, node subset $T \subseteq V$, 
and a subgraph $J=(V,E_J)$ of $G$ such that $T$ is $k$-connected in $J$,
find a minimum-cost augmenting edge-set $F \subseteq E \setminus E_J$ such that 
$T$ is $(k+1)$-connected in $J \cup F$.
The problem admits trivial ratio $O(|T|^2)$.
We consider the case $|T|>k$ and prove that for directed/undirected graphs and edge/node-costs,
a $\rho$-approximation for {\sf Rooted Subset $k$-Connectivity Augmentation} 
implies the following ratios for {\sf Subset $k$-Connectivity Augmentation}: \\
(i)  $b(\rho+k) + {\left(\frac{3|T|}{|T|-k}\right)}^2 H\left(\frac{3|T|}{|T|-k}\right)$; \\ 
(ii) $\rho \cdot O\left(\frac{|T|}{|T|-k} \log k \right)$, \\
where $b=1$ for undirected graphs and $b=2$ for directed graphs, 
and $H(k)$ is the $k$th harmonic number.
The best known values of $\rho$ on undirected graphs are 
$\min\{|T|,O(k)\}$ for edge-costs and 
$\min\{|T|,O(k \log |T|)\}$ for node-costs; 
for directed graphs $\rho=|T|$ for both versions.
Our results imply that unless $k=|T|-o(|T|)$, {\sf Subset $k$-Connectivity Augmentation} admits
the same ratios as the best known ones for the rooted version.
This improves the ratios in \cite{N-focs,L}.
\end{abstract}

% \noindent
% {\em Key words:} Graph-connectivity, Approximation algorithms. 

\section{Introduction} \label{s:intro}

% Let $\kappa_H(u,v)$ denote the maximaum number of internally-disjoint $uv$-paths in a graph $H$.
In the {\sf Survivable Network} problem we are given a graph $G=(V,E)$ with edge/node-costs 
and pairwise connectivity requirements $\{r(u,v):u,v \in T \subseteq V\}$
on a set $T$ of terminals.
The goal is to find a minimum-cost subgraph of $G$ that contains $r(u,v)$ internally-disjoint $uv$-paths
for all $u,v \in T$. 
In {\sf Rooted Subset $k$-Connectivity} problem there is $s \in T$ such that $r(s,t)=k$ for all 
$t \in T \setminus \{s\}$ and $r(u,v)=0$ otherwise.
In {\sf Subset $k$-Connectivity} problem $r(u,v)=k$ for all 
$u,v \in T$ and $r(u,v)=0$ otherwise.
In the {\em augmentation versions}, $G$ contains a subgraph $J$ of cost zero 
with $r(u,v)-1$ internally disjoint paths for all $u,v \in T$.
% such that $\kappa_J(u,v) \geq r(u,v)-1$. 
A subset $T \subseteq V$ of terminals is $k$-connected to a root $s$ in a 
directed/undirected graph $J$ if $J$ has $k$ internally-disjoint $vs$-paths 
for every $v \in T$; $T$ is $k$-connected in $J$ if $T$ is $k$-connected to every $s \in T$.
Formally, the versions of {\sf Survivable Network} we consider are as follows, 
where we revise our notation to $k \gets k+1$. 

\vspace{0.2cm}

\noindent
{\sf Rooted Subset $k$-Connectivity Augmentation} \\
{\em Instance:} \ 
A graph $G=(V,E)$ with edge/node-costs, a set $T \subseteq V$ of terminals, 
root $s \in T$, and a \hphantom{\em Instance: } subgraph $J=(V,E_J)$ of $G$ such that $T \setminus \{s\}$ 
 is $k$-connected to $s$ in $J$. \\
{\em Objective:}
Find a minimum-cost augmenting edge-set $F \subseteq E \setminus E_J$ such that 
$T \setminus \{s\}$ is $(k+1)$-connected \hphantom{\em Objective:} to $s$ in $J \cup F$.

\vspace{0.2cm}

\noindent
{\sf Subset $k$-Connectivity Augmentation} \\
{\em Instance:} \ 
A graph $G=(V,E)$ with edge/node-costs, subset $T \subseteq V$, and a subgraph  
 $J=(V,E_J)$ of \hphantom{\em Instance: } $G$ such that $T$ is $k$-connected in $J$. \\
{\em Objective:}
Find a minimum-cost augmenting edge-set $F \subseteq E \setminus E_J$ such that 
$T$ is $(k+1)$-connected \hphantom{\em Objective:} in $J \cup F$.

\vspace*{0.2cm}

The {\sf Subset $k$-Connectivity Augmentation} is {\sf Label-Cover} hard to approximate \cite{KKL}.
It is known and easy to see that for both edge-costs and node-costs,
if {\sf Subset $k$-Connectivity Augmentation} admits approximation ratio $\rho(k)$
such that $\rho(k)$ is a monotone increasing function,
then {\sf Subset $k$-Connectivity} admits ratio $k \cdot \rho(k)$.
Moreover, for edge costs, if in addition the approximation 
$\rho(k)$ is w.r.t. a standard setpair/biset LP-relaxation to the problem, 
then {\sf Subset $k$-Connectivity} admits ratio $H(k) \cdot \rho(k)$, where $H(k)$
denotes the $k$th harmonic number.
For edge-costs, a standard LP-relaxation for {\sf Survivable Network} (due to Frank and Jord\'{a}n \cite{FJ}) is:
$$\min\left\{\sum_{e \in E} c_e x_e : \sum_{e \in E(X,X^*)}x_e \geq r(X,X^*), X,X^* \subseteq V,
X \cap X^*=\emptyset, 0 \leq x_e \leq 1 \right\}$$
where $r(X,X^*)=\max\{r(u,v):u \in X, v \in X^*\}$ and $E(X,X^*)$ is the set of edges in $E$ from $X$ to $X^*$.

The {\sf Subset $k$-Connectivity} problem admits trivial ratios 
$O(|T|^2)$ for both edge-costs and node-costs, 
by computing for every $u,v \in V$ an optimal edge-set of 
$k$ internally-disjoint $uv$-paths (this is essentially a {\sf Min-Cost $k$-Flow} problem, 
that can be solved in polynomial time), and taking the union of the computed edge-sets.
We note that for metric edge-costs the problem admits an $O(1)$ ratio \cite{CV}.
For $|T| \geq k+1$ the problem can also be decomposed into $k$ instances of 
{\sf Rooted Subset $k$-Connectivity} problems,
c.f. \cite{KN2} for the case $T=V$, where it is also shown that for $T=V$ the number of  
of {\sf Rooted Subset $k$-Connectivity Augmentation} instances 
can be reduced to $O\left(\frac{|T|}{|T|-k} \log k\right)$,
which is $O(\log k)$ unless $k=|T|-o(|T|)$.

Recently, Laekhanukit \cite{L} made an important observation that the method of \cite{KN2}
can be extended for the case of arbitrary $T \subseteq V$.
Specifically, he proved that if $|T| \geq 2k$, then $O(\log k)$ instances 
of {\sf Rooted Subset $k$-Connectivity Augmentation} will suffice.
Thus for $|T| \geq 2k$, the $O(k)$-approximation algorithm of \cite{N-focs} 
for {\sf Rooted Subset $k$-Connectivity Augmentation}
leads to the ratio $O(k\log k)$ for {\sf Rooted Subset $k$-Connectivity Augmentation}.
By cleverly exploiting an additional property of the algorithm of \cite{N-focs} (see \cite[Lemma 14]{L}),
he reduced the ratio to $O(k)$ in the case $|T| \geq k^2$.

However, using a different approach, we will show that all this is not necessary, 
as for both directed and undirected graphs and edge-costs and node-costs, 
{\sf Subset $k$-Connectivity Augmentation} can be reduced to solving {\em one} 
instance (or two instances, in the case of directed graphs) of {\sf Rooted Subset $k$-Connectivity Augmentation}
and $O{\left(\frac{3|T|}{|T|-k}\right)}^2 H\left(\frac{3|T|}{|T|-k}\right)$ instances of 
{\sf Min-Cost $k$-Flow} problem. This leads to a much simpler algorithm, 
improves the result of Laekhanukit \cite{L} for $|T|<k^2$, and applies also for node-costs 
and directed graphs.
In addition, we give a more natural and much simpler extension of the algorithm of \cite{KN2} for $T=V$,
that also enables the same bound $O\left(\frac{|T|}{|T|-k} \log k\right)$ as in \cite{KN2} for arbitrary 
$T$ with $|T| \geq k+1$, and in addition applies also for directed graphs, for node-costs, 
and for an arbitrary type of edge-costs, e.g., metric costs, or uniform costs, or $0,1$-costs.
When we say ``$0,1$-edge-costs'' we mean that the input graph $G$ is complete, and the goal is 
to add to the subgraph $J$ of $G$ formed by the zero-cost edges a minimum size edge-set $F$ 
(any edge is allowed) such that $J \cup F$ satisfies the connectivity requirements.
Formally, our result is the following.

\begin{theorem} \label{t:SR}
For both directed and undirected graphs, and edge-costs and node-costs the following holds.
If {\sf Rooted Subset $k$-Connectivity Augmentation} admits approximation ratio 
$\rho=\rho(k,|T|)$, then for $|T| \geq k+1$
{\sf Subset $k$-Connectivity Augmentation} admits the following approximation ratios:
\begin{itemize}
\item[{\em (i)}]
$b(\rho+k)+{\left(\frac{|T|}{|T|-k}\right)}^2 O\left(\log \frac{|T|}{|T|-k}\right)$,
where $b=1$ for undirected graphs and $b=2$ for directed graphs.
\item[{\em (ii)}]
$\rho \cdot O\left(\frac{|T|}{|T|-k}  \log \min\{k,|T|-k\}\right)$, and this is so also for $0,1$-edge-costs.
\end{itemize}
Furthermore, if for edge-costs the approximation ratio $\rho$ is w.r.t. a standard
LP-relaxation for the problem, then so are the ratios in {\em (i)} and {\em (ii)}.
\end{theorem}

For $|T|>k$, the best known values of $\rho$ on undirected graphs are 
$O(k)$ for edge-costs and $\min\{O(k \log |T|),|T|\}$ for node-costs \cite{N-focs}; 
for directed graphs $\rho=|T|$ for both versions.
For $0,1$-edge-costs 
$\rho=O(\log k)$   \cite{N-aug} for undirected graphs and
$\rho=O(\log |T|)$ \cite{N-raug} for directed graphs.
For edge-costs, these ratios are w.r.t. a standard LP-relaxation.
Thus Theorem~\ref{t:SR} implies the following.

\begin{corollary} \label{c:main}
For $|T| \geq k+1$, {\sf Subset $k$-Connectivity Augmentation} admits the following approximation ratios.
\begin{itemize}
\item
For undirected graphs, the ratios are
$O(k)+{\left(\frac{|T|}{|T|-k}\right)}^2 O\left(\log \frac{|T|}{|T|-k}\right)$ for edge-costs,  
$O(k \log |T|)+{\left(\frac{|T|}{|T|-k}\right)}^2 O\left(\log \frac{|T|}{|T|-k}\right)$ for node-costs,
and $\frac{|T|}{|T|-k} \cdot O\left(\log^2k\right)$ for $0,1$-edge-costs.
\item
For directed graphs, the ratio is 
$2(|T|+k)+{\left(\frac{|T|}{|T|-k}\right)}^2 O\left(\log \frac{|T|}{|T|-k}\right)$
for both edge-costs and node-costs, and $\frac{|T|}{|T|-k} \cdot O\left(\log |T| \log k \right)$
for $0,1$ edge-costs.
\end{itemize}
For {\sf Subset $k$-Connecivity}, the ratios are larger by a factor of $H(k)$ for edge-costs,
and by a factor $k$ for node-costs.
\end{corollary}

Note that except the case of $0,1$-edge-costs,
Corollary~\ref{c:main} is deduced from part~(i) of Theorem~\ref{t:SR}.
However, part~(ii) of Theorem~\ref{t:SR} might become relevant if 
{\sf Rooted Subset $k$-Connectivity Augmentation} admits ratio 
better than $O(k)$. In addition, part~(ii) applies for {\em any type} of edge-costs, 
e.g. metric or $0,1$-edge-costs.

We conclude this section by mentioning some additional related work.
The case $T=V$ of {\sf Rooted Subset $k$-Connectivity} problem is the 
{\sf $k$-Outconnected Subgraph} problem; this problem admits a polynomial time algorithm for
directed graphs \cite{FT}, which implies ratio $2$ for undirected graphs.
For arbitrary $T$, the problem harder than {\sf Directed Steiner Tree} \cite{LN}.
The case $T=V$ of {\sf Subset $k$-Connectivity} problem is the 
{\sf $k$-Connected Subgraph} problem. This problem is NP-hard, 
and the best known ratio for it is $O\left(\log k \log \frac{n}{n-k}\right)$ for both
directed and undirected graphs \cite{N-DA9}; for the augmentation version of increasing the 
connectivity by one the ratio in \cite{N-DA9} is $O\left(\log \frac{n}{n-k}\right)$.
For metric costs the problem admits ratios $2+\frac{k-1}{n}$ for undirected graphs and $2+\frac{k}{n}$ 
for directed graphs \cite{KN1}.
For $0,1$-edge-costs the problem is solvable for directed graphs \cite{FJ}, which implies ratio $2$
for undirected graphs.
The {\sf Survivable Network} problem is {\sf Label-Cover} hard \cite{KKL}, and the currently best known 
non-trivial ratios for it on undirected graphs are: 
$O(k^3 \log |T|)$ for arbitrary edge-costs by Chuzhoy and Khanna \cite{CK-new}, 
$O(\log k)$ for metric costs due to Cheriyan and Vetta \cite{CV}, 
$O(k) \cdot \min\left\{\log^2 k, \log |T|\right\}$ for $0,1$-edge-costs \cite{N-aug,KN-aug},
and $O(k^4 \log^2 |T|)$ for node-costs \cite{N-focs}.

\section{Proof of Theorem~\ref{t:SR}} \label{s:SR}

We start by proving the following essentially known statement.

\begin{proposition} \label{p:connection}
Suppose that {\sf Rooted Subset $k$-Connectivity Augmentation} admits an approximation ratio $\rho$.
If for an instance of {\sf Subset $k$-Connectivity Augmentation} we are given a set of
$q$ edges (when any edge is allowed) and $p$ stars (directed to or from the root) on $T$ 
whose addition to $G$ makes $T$ $(k+1)$-connected,
then we can compute a $(\rho p+q)$-approximate solution $F$ to this instance in polynomial time. 
Furthermore, for edge-costs, if the $\rho$-approximation is w.r.t. 
a standard LP-relaxation, then $c(F) \leq (\rho p+q) \tau^*$, where $\tau^*$ is an optimal 
standard LP-relaxation value for {\sf Subset $k$-Connectivity Augmentation}.
\end{proposition}
\begin{proof}
For every edge $uv$ among the $q$ edges compute a minimum-cost edge-set $F_{uv} \subseteq E \setminus E_J$ 
such that $J \cup F_{uv}$ contains $k$ internally-disjoint $uv$-paths.
This can be done in polynomial time for both edge and node costs, using a 
{\sf Min-Cost $k$-Flow} algorithm. For edge-costs, it is known that $c(F_{uv}) \leq \tau^*$.
Then replace $uv$ by $F_{uv}$, and note that $T$ remains $k$-connected.
Similarly, for every star $S$ with center $s$ and leaf-set $T'$, compute an $\alpha$-approximate 
augmenting edge-set $F_S \subseteq E \setminus E_J$ such that 
$J \cup F_S$ contains $k$ internally-disjoint 
$sv$-paths (or $vs$-paths, in the case of directed graphs and $S$ being directed twords the root) 
for every $v \in T'$.
Then replace $S$ by $F_S$, and note that $T$ remains $k$-connected.
For edge-costs, it is known that if the $\rho$-approximation for the rooted version is w.r.t. 
a standard LP-relaxation, then $c(F_S) \leq (\alpha p+q) \tau^*$.
The statement follows.
\end{proof}

Motivated by Proposition~\ref{p:connection}, we consider the following question: \\
{\em Given a $k$-connected subset $T$ in a graph $J$, 
how many edges and/or stars on $T$ one needs to add to $J$ such that $T$ will become $(k+1)$-connected?} \\

We emphasize that we are interested in obtaining
{\em absolute bounds} on the number of edges in the question, expressed in certain parameters of the graph;
namely we consider the {\em extremal graph theory} question and not the {\em algorithmic problem}.
Indeed, the algorithmic problem of adding the minimum number of edges on $T$ 
such that $T$ will become $(k+1)$-connected can be shown to admit a polynomial-time algorithm
for directed graphs using the result of Frank and Jord\'{a}n \cite{FJ}; 
this also implies a $2$-approximation algorithm for undirected graphs.
However, in terms of the parameters $|T|,k$,
the result in \cite{FJ} implies only the trivial bound $O(|T|^2)$ on the the number of edges one needs to
add to $J$ such that $T$ will become $(k+1)$-connected.

Our bounds will be derived in terms of the family of the ``deficient'' sets of the graph $J$.
We need some definitions to state our results.

\begin{definition}
An ordered pair $\hat{X}=(X,X^+)$ of subsets of a groundset $V$ is
called a {\em biset} if $X \subseteq X^+$; $X$ is the {\em inner part} and 
$X^+$ is the {\em outer part} of $\hat{X}$,
$\Gamma(\hat{X})= X^+ \setminus X$ is the boundary of $\hat{X}$,
and $X^*=V \setminus X^+$ is the {\em complementary set} of $\hat{X}$.
% The {\em projection} of a biset $\hat{X}=(X,X^+)$ on a subset $T \subseteq V$ of terminals 
% is the biset $\hat{X}_T=(X \cap T,X^+ \cap T)$.
\end{definition}

Given an instance of {\sf Subset $k$-Connectivity Augmentation} we may assume that $T$ is an independent set in $J$.
Otherwise, we obtain an equivalent instance by subdividing every edge $uv \in J$ with $u,v \in T$ by a new node.

\begin{definition}
Given a $k$-connected independent set $T$ in a graph $J=(V,E_J)$ 
let us say that a biset $\hat{X}$ on $V$ is {\em $(T,k)$-tight} in $J$ 
if $X \cap T,X^* \cap T \neq \emptyset$,
$X^+$ is the union of $X$ and the set of neighbors of $X$ in $J$, and $|\Gamma(\hat{X})|=k$.
\end{definition}

An edge covers a biset $\hat{X}$ if it goes from $X$ to $X^*$.
By Menger's Theorem, $F$ is a feasible solution to {\sf Subset $k$-Connectivity Augmentation} if, and only if,
$F$ covers the biset-family ${\cal F}$ of tight bisets; see \cite{KN-sur,N-aug}.
Thus our question can be reformulated as follows: \\
{\em Given a $k$-connected independent set $T$ in a graph $J$, 
how many edges and/or stars on $T$ are needed to cover the family ${\cal F}$ of $(T,k)$-tight bisets?}

\begin{definition} \label{d:main}
The intersection and the union of two bisets $\hat{X},\hat{Y}$ is defined by 
$\hat{X} \cap \hat{Y} = (X \cap Y,X^+ \cap Y^+)$ and 
$\hat{X} \cup  \hat{Y} = (X \cup Y,X^+ \cup Y^+)$. 
% We say that two bisets $\hat{X},\hat{Y}$ {\em cross} if $X \cap Y, X^* \cap Y^* \neq \emptyset$.
Two bisets $\hat{X},\hat{Y}$ {\em intersect} if $X \cap Y \neq \emptyset$;
if in addition $X^* \cap Y^* \neq \emptyset$ then $\hat{X},\hat{Y}$ {\em cross}.
We say that a biset-family ${\cal F}$ is:
\begin{itemize}
\item
{\em crossing} if $\hat{X} \cap \hat{Y},\hat{X} \cup \hat{Y} \in {\cal F}$ 
for any $\hat{X},\hat{Y} \in {\cal F}$ that cross.
\item
{\em $k$-regular} if $|\Gamma(\hat{X})| \leq k$ for every $\hat{X} \in {\cal F}$, and if  
$\hat{X} \cap \hat{Y},\hat{X} \cup \hat{Y} \in {\cal F}$ 
for any intersecting $\hat{X},\hat{Y} \in {\cal F}$
with $|X \cup Y| \leq |T|-k-1$.
\end{itemize}
\end{definition}

The following statement is essentially known. 

\begin{lemma} \label{l:properties}
Let $T$ be a $k$-connected independent set in a graph $J=(V,E_J)$,
and let $\hat{X},\hat{Y}$ be $(T,k)$-tight bisets. 
If $(X \cap T,X^+ \cap T),(Y \cap T,Y^+ \cap T)$ cross or if $|(X \cup Y) \cap T| \leq |T|-k-1$
then $\hat{X} \cap \hat{Y},\hat{X} \cup \hat{Y}$ are both $(T,k)$-tight. 
\end{lemma}
\begin{proof}
The case $(X \cap T,X^+ \cap T),(Y \cap T,Y^+ \cap T)$
was proved in \cite{N-aug} and \cite{L}.
The proof of the case $|(X \cup Y) \cap T| \leq |T|-k-1$
is identical to the proof of \cite[Lemma~1.2]{J} where the case $T=V$ is considered.
\end{proof}

\begin{corollary}
The biset-family 
$${\cal F}=\{(X \cap T,X^+\cap T): (X,X^+) \mbox{ is a } (T,k)\mbox{-tight biset in } J\}$$
is crossing and $k$-regular, and the {\em reverse family} 
$\bar{\cal F}=\{(T \setminus X^+,T \setminus X): \hat{X} \in {\cal F}\}$
of ${\cal F}$ is also crossing and $k$-regular.
Furthermore, if $J$ is undirected then ${\cal F}$ is symmetric, namely,
${\cal F}=\bar{\cal F}$.
% $(X,X^+) \in {\cal F}$ implies $(V \setminus X^+,V \setminus X) \in {\cal F}$.
\end{corollary}

Given two bisets $\hat{X},\hat{Y}$ we write $\hat{X} \subseteq \hat{Y}$ and say that 
$\hat{Y}$ contains $\hat{X}$ if $X \subseteq Y$ or if $X=Y$ and $X^+ \subseteq Y^+$;
$\hat{X} \subset \hat{Y}$ and $\hat{Y}$ properly contains $\hat{X}$ if
$X \subset Y$ or if $X=Y$ and $X^+ \subset Y^+$. 

\begin{definition}
A biset $\hat{C}$ is a {\em core} of a biset-family ${\cal F}$ if $\hat{C} \in {\cal F}$
and $\hat{C}$ contains no biset in ${\cal F} \setminus \{\hat{C}\}$;
namely, a core is an inclusion-minimal biset in ${\cal F}$.
Let ${\cal C}({\cal F})$ be the family of cores of ${\cal F}$ and let 
$\nu({\cal F})=|{\cal C}({\cal F})|$ denote their number.
\end{definition}

Given a biset-family ${\cal F}$ and an edge-set $I$ on $T$,
the residual biset-family ${\cal F}_I$ of ${\cal F}$ consists of the members of ${\cal F}$ uncovered by $I$.
We will assume that for any $I$, the cores of ${\cal F}_I$ and of $\bar{\cal F}_I$
can be computed in polynomial time. 
For ${\cal F}$ being the family of $(T,k)$-tight bisets this can be implemented in polynomial time
using the Ford-Fulkerson {\sf Max-Flow Min-Cut} algorithm, c.f. \cite{N-aug}.
It is known and easy to see that if ${\cal F}$ is crossing 
and/or $k$-regular, so is ${\cal F}_I$, for any edge-set $I$. 

\begin{definition}
% A biset-family ${\cal F}$ on $T$ is {\em intersection-closed} if $\hat{X} \cap \hat{Y} \in {\cal F}$ 
% for any $\hat{X},\hat{Y} \in {\cal F}$ that intersect.
For a biset-family ${\cal F}$ on $T$ let $\nu({\cal F})$ be the maximum number of bisets in ${\cal F}$ 
which inner parts are pairwise-disjoint.
For an integer $k$ let ${\cal F}^k=\{\hat{X} \in {\cal F}:|X| \leq (|T|-k)/2\}$.
\end{definition}

\begin{lemma} \label{l:F*}
Let ${\cal F}$ be a $k$-regular biset-family on $T$ and let $\hat{X}, \hat{Y} \in {\cal F}^k$ intersect.
Then $\hat{X} \cap \hat{Y} \in {\cal F}^k$ and $\hat{X} \cup \hat{Y} \in {\cal F}$.
\end{lemma}
\begin{proof}
Since $|X|,|Y| \leq \frac{|T|-k}{2}$, we have
$|X \cup Y| = |X|+|Y|-|X \cap Y| \leq |T|-k-1$.
Thus $\hat{X} \cap \hat{Y}, \hat{X} \cap \hat{Y} \in {\cal F}$, by the $k$-regularity of ${\cal F}$.
Moreover, $\hat{X} \cap \hat{Y} \in {\cal F}^k$, since $|X \cap Y| \leq |X| \leq \frac{|T|-k}{2}$.
\end{proof}

We will prove the following two theorems that imply Theorem~\ref{t:SR}.

\begin{theorem} \label{t:edges}
Let ${\cal F}$ be a biset-family on $T$ such that both ${\cal F},\bar{\cal F}$ are crossing and $k$-regular.
Then there exists a polynomial-time algorithm that computes an edge-cover $I$ of ${\cal F}$ of size 
$|I|=\nu\left({\cal F}^k\right)+\nu\left(\bar{\cal F}^k\right)+
     {\left(\frac{3|T|}{|T|-k}\right)}^2 H\left(\frac{3|T|}{|T|-k}\right)$.
Furthermore, if ${\cal F}$ is symmetric then 
$|I|=\nu\left({\cal F}^k\right)+{\left(\frac{3|T|}{|T|-k}\right)}^2 H\left(\frac{3|T|}{|T|-k}\right)$.
\end{theorem}

\begin{theorem} \label{t:stars}
Let ${\cal F}$ be a biset-family on $T$ such that both ${\cal F}$ and $\bar{\cal F}$ are $k$-regular.
Then there exists a collection of $O\left(\frac{|T|}{|T|-k} \lg \min\{\nu,|T|-k\}\right)$ stars on $T$ which union
covers ${\cal F}$, and such a collection can be computed in polynomial time.
Furthermore, the total number of edges in the stars is at most 
$\nu\left({\cal F}^k\right) + \nu\left(\bar{\cal F}^k\right) +
{\left(\frac{|T|}{|T|-k}\right)}^2 \cdot O\left(\log \frac{|T|}{|T|-k}\right)$.
\end{theorem}

Note that the second statement in Theorem~\ref{t:stars} implies (up to constants)
the bound in Theorem~\ref{t:edges}.
However, the proof of Theorem~\ref{t:edges} is much simpler than the proof of Theorem~\ref{t:stars},
and the proof of Theorem~\ref{t:edges} is a part of the proof of 
the second statement in Theorem~\ref{t:stars}.

Let us show that Theorems \ref{t:edges} and \ref{t:stars} imply Theorem~\ref{t:SR}.
For that, all we need is to show that by applying one time the $\alpha$-approximation algorithm 
for the {\sf Rooted Subset $k$-Connectivity Augmentation}, we obtain an instance with 
$\nu\left({\cal F}^k\right),\nu\left(\bar{\cal F}^k\right) \leq k+1$. 
This is achieved by the following procedure due to 
Khuller and Raghavachari \cite{KR} that originally considered the case $T=V$, see also \cite{ADNP,DN,KN1};
the same procedure is also used by Laekhanukit in \cite{L}.

Choose an arbitrary subset $T' \subseteq T$ of $k+1$ nodes, add a new node $s$ (the root) 
and all edges between $s$ and $T'$ of cost zero each, both to $G$ and to $J$. 
Then, using the $\alpha$-approximation algorithm for the 
{\sf Rooted Subset $k$-Connectivity Augmentation}, compute an augmenting edge set $F$  
such that $J \cup F$ contains $k$ internally disjoint $vs$-paths and $sv$-paths for every $v \in T'$.
Now, add $F$ to $J$ and remove $s$ from $J$.
It is a routine to show that $c(F) \leq b {\sf opt}$, and that for edge-costs 
$c(F) \leq b \tau^*$. It is also known that if $\hat{X}$ is a tight biset of the obtained graph $J$,
then $X \cap T', X^* \cap T' \neq \emptyset$, c.f. \cite{ADNP,L}.
Combined with Lemma~\ref{l:F*} we obtain that 
$\nu\left({\cal F}^k\right),\nu\left(\bar{\cal F}^k\right) \leq |T'| \leq k+1$
for the obtained instance, as claimed.

\section{Proof of Theorem \ref{t:edges}} \label{s:edges}

\begin{definition}
Given a biset-family ${\cal F}$ on $T$, let 
$\Delta({\cal F})$ denote the maximum degree in the hypergraph 
${\cal F}^{in}=\{X:\hat{X} \in {\cal F}\}$ of the inner parts of the bisets in ${\cal F}$.
We say that $T' \subseteq T$ is a {\em transversal} of ${\cal F}$ if 
$T' \cap X \neq \emptyset$ for every $X \in {\cal F}^{in}$; a function 
$t:T \rightarrow [0,1]$ is a {\em fractional transversal} of ${\cal F}$ if 
$\sum_{v \in X} t(v) \geq 1$ for every $X \in {\cal F}^{in}$.
\end{definition}

\begin{lemma} \label{l:Delta}
Let ${\cal F}$ be a crossing biset-family. Then
$\Delta({\cal C}({\cal F})) \leq \nu\left(\bar{\cal F}\right)$.
\end{lemma}
\begin{proof}
Since ${\cal F}$ is crossing, the members of ${\cal C}({\cal F})$ are pairwise non-crossing.
Thus if ${\cal H}$ is a subfamily of ${\cal C}({\cal F})$ 
such that the intersection of the inner parts of the bisets in ${\cal H}$ is non-empty,
then $\bar{\cal H}$ is a subfamily of $\bar{\cal F}$ 
such that the inner parts of the bisets in $\bar{\cal H}$ are pairwise disjoint,
so $|\bar{\cal H}| \leq \nu\left(\bar{\cal F}\right)$.
The statement follows.
\end{proof}

\begin{lemma} \label{l:final-cover}
Let $T'$ be a transversal of a biset-family ${\cal F}'$ on $T$ and let $I'$ be an edge-set on $T$ 
obtained by picking for every $s \in T'$ an edge from $s$ to every inclusion member 
of the set-family $\{X^*:\hat{X} \in {\cal F}', s \in X\}$. Then $I'$ covers ${\cal F}'$.
Moreover, if ${\cal F}'$ is crossing then $|I'| \leq |T'| \cdot \nu(\bar{\cal F}')$.
\end{lemma}
\begin{proof}
The statement that $I'$ covers ${\cal F}'$ is obvious. 
If ${\cal F}'$ is crossing, then for every $s \in T$ the inclusion-minimal 
members of $\{X^*:\hat{X} \in {\cal F}', s \in X\}$ are pairwise-disjoint,
hence their number is at most $\nu(\bar{\cal F}')$. The statement follows.
\end{proof}

\begin{lemma} \label{l:bounds}
Let ${\cal F}$ be a $k$-regular biset-family on $T$. Then the following holds.
\begin{itemize}
\item[{\em (i)}]
$\nu({\cal F}) \leq \nu\left({\cal F}^k\right)+\frac{2|T|}{|T|-k}$. 
\item[{\em (ii)}]
If $\nu\left({\cal F}^k_{\{e\}}\right) = \nu\left({\cal F}^k\right)$ holds for every edge $e$ on $T$
then $\nu\left({\cal F}^k\right) \leq \frac{|T|}{|T|-k}$.
\item[{\em (iii)}]
There exists a polynomial time algorithm that finds 
a transversal $T'$ of ${\cal C}({\cal F})$ of size at most
$|T'| \leq \left(\nu\left({\cal F}^k\right)+\frac{2|T|}{|T|-k}\right) \cdot H(\Delta({\cal C}({\cal F})))$.
\end{itemize}
\end{lemma}
\begin{proof}
Part~(i) is immediate.

We prove (ii). Let $\hat{C} \in {\cal C}\left({\cal F}^k\right)$ and let $\hat{U}_C$
be the union of the bisets in ${\cal F}^k$ that contain $\hat{C}$ and contain no other 
member of ${\cal C}\left({\cal F}^k\right)$. If $|U_C| \leq |T|-k-1$ then 
$\hat{U}_C \in {\cal F}$, by the $k$-regularity of ${\cal F}$.
In this case $\nu\left({\cal F}^k_{\{e\}}\right) \leq \nu\left({\cal F}^k\right)-1$
for any edge from $C$ to $U_C^*$. Hence $|U_C| \geq |T|-k$ must hold for every $\hat{C} \in {\cal C}({\cal F})$.
By Lemma~\ref{l:F*}, the sets in the set family $\{U_C: \hat{C} \in {\cal C}({\cal F})\}$ are pairwise disjoint.
The statement follows.

We prove (iii). % We claim that ${\cal F}$ has a fractional transversal of value $\mu^k({\cal F})$.
Let $T^k$ be an inclusion-minimal transversal of ${\cal F}^k$.
By Lemma~\ref{l:F*}, $\left|T^k\right|=\nu\left({\cal F}^k\right)$.
Setting $t(v)=1$ if $v \in T^k$ and $t(v)=\frac{2}{|T|-k}$ otherwise,
we obtain a fractional transversal of ${\cal C}({\cal F})$ of value at most 
$\nu\left({\cal F}^k\right)+\frac{2|T|}{|T|-k}$.
Consequently, the greedy algorithm of Lov\'{a}sz \cite{Lov} finds a transversal $T'$ as claimed. 
\end{proof}

The algorithm for computing $I$ as in Theorem~\ref{t:edges} starts with $I=\emptyset$ and then
continues as follows. 

\vspace{0.3cm} 

\noindent
{\bf Phase~1} \\
While there exists an edge $e$ on $T$ 
such that $\nu\left({\cal F}^k_{I \cup \{e\}}\right) \leq \nu\left({\cal F}^k_I\right)-1$, or 
such that $\nu\left(\bar{\cal F}^k_{I \cup \{e\}}\right) \leq \nu\left(\bar{\cal F}^k_I\right)-1$,
add $e$ to $I$. \\

\noindent
{\bf Phase~2} \\
Find a transversal $T'$ of ${\cal C}({\cal F}')$ as in Lemma~\ref{l:bounds}(iii),
where ${\cal F}'={\cal F}_I$.
Then find an edge-cover $I'$ of ${\cal F}'$ as in Lemma~\ref{l:final-cover} and add $I'$ to $I$.

\vspace{0.3cm}

The edge-set $I$ computed covers ${\cal F}$ by Lemma~\ref{l:final-cover}.
Clearly, the number of edges in $I$ at the end of Phase~1 is at most 
$\nu\left({\cal F}^k\right)+\nu\left(\bar{\cal F}^k\right)$, and is at most 
$\nu\left({\cal F}^k\right)$ if ${\cal F}$ is symmetric.
Now we bound the size of $I'$. 
Note that at the end of Phase~1 we have 
$\nu\left({\cal F}_I^k\right), \nu\left(\bar{\cal F}_I^k\right) \leq \frac{|T|}{|T|-k}$ (by Lemma~\ref{l:bounds}(ii))
and thus $\nu\left(\bar{\cal F}_I\right) \leq \frac{3|T|}{|T|-k}$ (by Lemma~\ref{l:bounds}(i)) and 
$\Delta({\cal C}({\cal F}_I)) \leq \nu\left(\bar{\cal F}_I\right) \leq 
\nu\left(\bar{\cal F}_I^k\right)+\frac{2|T|}{|T|-k} \leq \frac{3|T|}{|T|-k}$
(by Lemma~\ref{l:Delta}).
Consequently, 
$|T'| \leq \left(\nu\left({\cal F}_I^k\right)+\frac{2|T|}{|T|-k}\right) \cdot H(\Delta({\cal C}({\cal F}_I))) \leq 
\frac{3|T|}{|T|-k} \cdot H\left( \frac{3|T|}{|T|-k}\right)$.
From this we get
$|I'| \leq |T'| \cdot \nu\left(\bar{\cal F}_I\right) \leq 
{\left(\frac{3|T|}{|T|-k}\right)}^2 \cdot H\left( \frac{3|T|}{|T|-k}\right)$.

The proof of Theorem~\ref{t:edges} is now complete.

\section{Proof of Theorem~\ref{t:stars}} \label{s:stars}

We start by analyzing the performance of a natural {\sf Greedy Algorithm} for covering
$\nu\left({\cal F}^k\right)$, that starts with $I=\emptyset$ and while $\nu({\cal F}^k_I) \geq 1$
adds to $I$ a star $S$ for which $\nu({\cal F}^k_{I \cup S})$ is minimal.
It is easy to see that the algorithm terminates since any star with center $s$ 
in the inner part of some core of ${\cal F}^k_I$ and edge set $\{vs:v \in T \setminus \{s\}\}$
reduces the number of cores by one. The proof of the following statement is similar to the proof
of the main result of \cite{KN2}.

\begin{lemma} \label{l:greedy}
Let ${\cal F}$ be a $k$-regular biset-family and let
${\cal S}$ be the collection of stars computed by the {\sf Greedy Algorithm}. Then
$$
|{\cal S}|= O\left(\frac{|T|}{|T|-k}  \ln \min\left\{\nu\left({\cal F}^k\right),|T|-k\right\}\right)
\ .$$
\end{lemma}

Recall that given $\hat{C} \in {\cal C}\left({\cal F}^k_I\right)$ we denote by $\hat{U}_C$
the union of the bisets in ${\cal F}^k_I$ that contain $\hat{C}$ and contain no other 
member of ${\cal C}\left({\cal F}^k_I\right)$, and that 
by Lemma~\ref{l:F*}, the sets in the set-family $\{U_C: \hat{C} \in {\cal C}({\cal F})\}$ are pairwise disjoint.

\begin{definition} [\cite{KN2}]
Let us say that $s \in V$ {\em out-covers} $\hat{C} \in {\cal C}\left({\cal F}^k\right)$ if $s \in U^*_C$.
\end{definition}

\begin{lemma} \label{l:cover}
Let ${\cal F}$ be $k$-regular biset-family and let $\nu=\nu\left({\cal F}^k\right)$.
\begin{itemize}
\item[{\em (i)}]
There is $s \in T$ that out-covers at least $\nu\left(1-\frac{k}{|T|}\right)-1$
members of ${\cal C}\left({\cal F}^k\right)$.
\item[{\em (ii)}]
Let $s$ out-cover the members of ${\cal C} \subseteq {\cal C}\left({\cal F}^k\right)$ 
and let $S$ be a star with one edge from $s$ to the inner part of each member of ${\cal C}$.
Then $\nu({\cal F}^k) \leq \nu({\cal F}^k_S)-|{\cal C}|/2$.
\end{itemize}
Consequently, there exists a star $S$ on $T$ such that 
\begin{equation} \label{e:s}
\nu({\cal F}^k_S) \leq \frac{1}{2}\left(1+\frac{k}{|T|}\right) \cdot \nu +\frac{1}{2} = 
\alpha \cdot \nu +\beta \ .
\end{equation}
\end{lemma}
\begin{proof}
We prove (i). Consider the hypergraph 
${\cal H}=\left\{T \setminus \Gamma\left(\hat{U}_C\right): \hat{C} \in {\cal C}\left({\cal F}^k\right)\right\}$.
Note that the number of members of ${\cal C}\left({\cal F}^k\right)$ out-covered by any 
$v \in T$ is at least the degree of $s$ in ${\cal H}$ minus $1$.
Thus all we need to prove is that there is a node $s \in T$ whose degree 
in ${\cal H}$ is at least $\nu\left(1-\frac{k}{|T|}\right)$.
For every $C \in {\cal C}({\cal F})$ we have 
$\left|T \setminus \Gamma\left(\hat{U}_C\right)\right| \geq |T|-k$, by the $k$-regularity of ${\cal F}$.
Hence the bipartite incidence graph of ${\cal H}$ 
has at least $\nu(|T|-k)$ edges, and thus has a node $s \in T$ of degree at least
$\nu\left(1-\frac{k}{|T|}\right)$, which equals the degree of $s$ in ${\cal H}$.
Part (i) follows.

We prove (ii).
It is sufficient to show that every $\hat{C} \in {\cal C}\left({\cal F}^k_S\right)$ contains some
$\hat{C}' \in {\cal C}\left({\cal F}^k\right) \setminus {\cal C}$ or contains at least two members in ${\cal C}$.
% Let $\hat{C'} \in {\cal C}\left({\cal F}^k_S\right)$. 
Clearly, $\hat{C}$ contains some $\hat{C}' \in {\cal C}\left({\cal F}^k\right)$.
We claim that if $\hat{C}' \in {\cal C}$ then $\hat{C}$ must contain 
some $\hat{C}'' \in {\cal C}\left({\cal F}^k\right)$ distinct from $\hat{C}'$. 
Otherwise, $\hat{C} \in {\cal F}^k(C)$. But as $S$ covers all members
of ${\cal F}^k(C)$, $\hat{C} \notin {\cal F}^k_S$. This is a contradiction.
\end{proof}

Let us use parameters $\alpha,\beta,\gamma,\delta$ and $j$ set to
$$\alpha = \frac{1}{2}\left(1+\frac{k}{|T|}\right)  \ \ \ \ \   
  \beta  = \frac{1}{2} \ \ \ \ \  
  \gamma = 1-\frac{k}{|T|}=2(1-\alpha) \ \ \ \ \   
  \delta = 1,
$$ 
and $j$ is the minimum integer such that 
$\alpha^j\left(\nu-\frac{\beta}{1-\alpha}\right) \leq \frac{2}{1-\alpha}$ (note that $\alpha<1$), namely,   
\begin{equation} \label{e:j}
j=\left\lfloor \frac{\ln \frac{1}{2}(\nu(1-\alpha)-\beta)}{\ln(1/\alpha)}\right\rfloor \leq 
\left\lfloor \frac{\ln \frac{1}{2}\nu(1-\alpha)}{\ln(1/\alpha)}\right\rfloor \ .
\end{equation}
We assume that $\nu \geq  \frac{2+\beta}{1-\alpha}$ to have $j \geq 0$ 
(otherwise Lemma~\ref{l:greedy} follows). Note that $\frac{\beta}{1-\alpha}=\frac{|T|}{|T|-k}$.

\begin{lemma} \label{l:computations}
Let $0 \leq \alpha < 1$, $\beta \geq 0$, $\nu_0=\nu$, and  for $i \geq 1$ let 
$$\nu_{i+1} \leq \alpha \nu_i + \beta \ \ \ \ \ s_i = \gamma \nu_{i-1}-\delta \ .$$
Then 
$\nu_i \leq \alpha^i\left(\nu-\frac{\beta}{1-\alpha}\right) + \frac{\beta}{1-\alpha}$
and
$\sum\limits_{i=1}^j s_i \leq 
\frac{1-\alpha^j}{1-\alpha} \cdot \gamma\left(\nu-\frac{\beta}{1-\alpha}\right) + 
j\left(\frac{\gamma \beta}{1-\alpha}-\delta\right)
$.
Moreover, if $j$ is given by (\ref{e:j}) then 
$\nu_j \leq \frac{2+\beta}{1-\alpha}=\frac{5|T|}{|T|-k}$ and 
$\sum_{i=1}^j s_i \leq  2\left(\nu-\frac{|T|}{|T|-k}\right)$.
\end{lemma}
\begin{proof}
Unraveling the recursive inequality $\nu_{i+1} \leq \alpha \nu_i + \beta$ in the lemma we get:
$$\nu_i \leq \alpha^i \nu +\beta\left(1+\alpha+ \cdots +\alpha^{i-1}\right) = 
\alpha^i \nu +\beta \frac{1-\alpha^i}{1-\alpha} = 
\alpha^i\left(\nu-\frac{\beta}{1-\alpha}\right) + \frac{\beta}{1-\alpha} \ .$$
This implies 
$s_i \leq \gamma\left(\nu-\frac{\beta}{1-\alpha}\right) \alpha^{i-1} + \frac{\gamma \beta}{1-\alpha}-\delta$,
and thus 
\begin{eqnarray*}
\sum_{i=1}^j s_i & \leq & \gamma\left(\nu-\frac{\beta}{1-\alpha}\right) \sum_{i=1}^j \alpha^{i-1} + 
                          j\left(\frac{\gamma \beta}{1-\alpha}-\delta\right) \\
                 & =    & \gamma\left(\nu-\frac{\beta}{1-\alpha}\right) \cdot \frac{1-\alpha^j}{1-\alpha} + 
                          j\left(\frac{\gamma \beta}{1-\alpha}-\delta\right)
\end{eqnarray*}
If $j$ is given by (\ref{e:j}) then 
$\nu_j \leq \alpha^i\left(\nu-\frac{\beta}{1-\alpha}\right) + \frac{\beta}{1-\alpha} \leq 
\frac{2}{1-\alpha} + \frac{\beta}{1-\alpha} = \frac{2+\beta}{1-\alpha}$, and
\begin{eqnarray*}
\sum\limits_{i=1}^j s_i & \leq & 
\frac{1-\alpha^j}{1-\alpha} \cdot \gamma\left(\nu-\frac{\beta}{1-\alpha}\right) + 
j\left(\frac{\gamma \beta}{1-\alpha}-\delta\right) \\
& \leq &  
2\left(\nu-\frac{\beta}{1-\alpha}\right) = 2\left(\nu-\frac{|T|}{|T|-k}\right) \ .
\end{eqnarray*}
\end{proof}

We now finish the proof of Lemma~\ref{l:greedy}.
At each one of the first $j$ iterations we out-cover at least 
$\nu\left({\cal F}^k_I\right)\left(1-\frac{k}{|T|}\right)-1$
members of ${\cal C}\left({\cal F}^k_I\right)$, by Lemmas \ref{l:cover}.
In each one of the consequent iterations, we can reduce $\nu\left({\cal F}^k_I\right)$
by at least one, if we choose the center of the star in $C$ for some $\hat{C} \in {\cal C}\left({\cal F}^k_I\right)$.
Thus using Lemma~\ref{l:computations}, performing the necessary computations, and substituting 
the values of the parameters, we obtain that the number of stars in ${\cal S}$ is bounded by
$$j+\nu_j \leq \left\lfloor \frac{\ln \frac{1}{2}\nu(1-\alpha)}{\ln(1/\alpha)}\right\rfloor + \frac{5|T|}{|T|-k} =
O\left(\frac{|T|}{|T|-k}  \ln \min\{\nu,|T|-k\}\right) \ .$$

Now we discuss a variation of this algorithm that produces ${\cal S}$ with a small number of leaves.
Here at each one of the first $j$ iterations we out-cover {\em exactly} $\nu\left(1-\frac{k}{|T|}\right)-1$
min-cores. For that, we need be able to compute the bisets $\hat{U}_C$, and such a procedure
can be found in \cite{L}. 
The number of edges in the stars at the end of this phase is 
at most $2\left(\nu-\frac{|T|}{|T|-k}\right)$ and $\nu_j \leq \frac{5|T|}{|T|-k}$.
In the case of non-symmetric ${\cal F}$ and/or directed edges, we apply the same algorithm on 
$\bar{\cal F}^k$.
At this point, we apply Phase~2 of the algorithm from the previous section.
Since the number of cores of each one of ${\cal F}^k_I,\bar{\cal F}^k_I$ is now $O\left(\frac{|T|}{|T|-k}\right)$,
the size of the transversal $T'$ computed is bounded by 
$|T'|=O\left(\frac{|T|}{|T|-k} \cdot \log \frac{|T|}{|T|-k}\right)$.
The number of stars is at most the size $|T'|$, while the
number of edges in the stars is at most 
$|T'| \cdot \nu\left(\bar{\cal F}_I\right)=
{\left(\frac{|T|}{|T|-k}\right)}^2 \cdot O\left(\log \frac{|T|}{|T|-k}\right)$.

This concludes the proof of Theorem~\ref{t:stars}.

% \section{Conclusions}

% \bibliographystyle{abbrv}
% \bibliography{subs-conn}

\end{document}